\documentclass[lettersize,journal]{IEEEtran}
\usepackage{amsmath,amsfonts}
\usepackage{algorithmic}
\usepackage{array}
\usepackage[caption=false,font=normalsize,labelfont=sf,textfont=sf]{subfig}
\usepackage{textcomp}
\usepackage{stfloats}
\usepackage{url}
\usepackage{verbatim}
\usepackage{graphicx}
\usepackage{amsmath}
\usepackage[ruled]{algorithm2e}
\usepackage{amsthm}
\usepackage{amssymb}
\newtheorem{definition}{Definition}
\newtheorem{theorem}{Theorem}
\newtheorem{lemma}{Lemma}
\newtheorem{proposition}{Proposition}
\newtheorem{remark}{Remark}
\def\BibTeX{{\rm B\kern-.05em{\sc i\kern-.025em b}\kern-.08em
		T\kern-.1667em\lower.7ex\hbox{E}\kern-.125emX}}
\usepackage{balance}

\begin{document}
	\title{From Target Tracking to Targeting Track — Part II: Regularized Polynomial Trajectory Optimization}
	\author{Tiancheng Li, \textit{IEEE Senior Member}, Yan Song, Guchong Li, Hao Li
		\thanks{Manuscript created August 2024; \\
			This work was supported in part by the National Natural Science Foundation of China under Grants 62422117 and 62201316 
			and in part by the Fundamental Research Funds for the Central Universities. 
			\\
			Tiancheng Li, Yan Song, Guchong Li and H. Li are with the Key Laboratory of Information Fusion Technology (Ministry of Education), School of Automation, Northwestern Polytechnical University, Xi’an 710129, China, E-mail: t.c.li@nwpu.edu.cn, syzx@mail.nwpu.edu.cn, guchong.li@nwpu.edu.cn, lihao714925@163.com. 
			Hao Li is also with Xi'an Branch of China Academy of Space Technology, Xi'an 710100, China
	}}
	
	\markboth{Journal of \LaTeX\ Class Files, August~2024}%
	{How to Use the IEEEtran \LaTeX \ Templates}
	
	\maketitle
	
	\begin{abstract}
		Target tracking entails the estimation of the evolution of the target state over time, namely the target trajectory. Different from the classical state space model, our series of studies, including this paper, model the collection of the target state as a stochastic process (SP) that is further decomposed into a deterministic part which represents the trend of the trajectory and a residual SP representing the residual fitting error. 
		Subsequently, 
		the tracking problem is formulated as a learning task regarding the trajectory SP 
		for which a key part is to estimate a trajectory FoT (T-FoT) best fitting the measurements in time series. For this purpose, we consider the polynomial curve and address the regularized polynomial T-FoT optimization employing two distinct regularization strategies seeking trade-off between the accuracy and simplicity. 
		One limits the order of the polynomial and then the best choice is determined by grid searching in a narrow, bounded range while the other adopts $\ell_0$ norm regularization for which the hybrid Newton solver is employed. 
		Simulation results obtained in both single and multiple maneuvering target scenarios demonstrate the effectiveness of our approaches. 
	\end{abstract}
	

	\section{Introduction}\label{sec:Introduction}
	\IEEEPARstart{T}{he} online estimation of the trajectory of a moving target such as airplane, robot, missile, etc., called \textit{target tracking}, has been a persistent and prominent research topic, with wide-ranging implications in aerospace, traffic management and defense, among many others \cite{bar2004estimation,vo2015multitarget}. 
	The classic methodology, exemplified by a wide variety of Bayesian filters, is to design a state space model (SSM) \cite{Sarkka13book} consisting of a Markov-jump model to describe the dynamics of the target and a measurement model to relate the measurement with the state of the target. The SSM-based methods can be divided into two major classes, depending on whether the ground truth state is assumed to be a deterministic or random variable, which leads to two distinctive groups of estimators, as shown on the left part of Fig.\ref{fig:Taxonomy}. One is the minimum mean square error (MMSE) driven point state estimator such as the milestone Kalman filter, 
	Gaussian mixture filter, 
	and so on \cite{Li2017AGC}. The other is the Bayesian risk-driven density estimator  such as the particle filter and various random finite set filters \cite{vo2015multitarget}. 
	
	
	\subsection{Relevant Work}
	The actual motion modes are typically unknown in practice and even time-variant, that is, maneuvering \cite{li2005survey}, or simply are too complicated to be sufficiently described \cite{Zhang24L-GBM}. In addition, appropriate models are needed to characterize false, missing and even irregular data \cite{bar2004estimation,Sarkka13book}, which is also intractable. 
	In the absence of accurate a priori knowledge, model mismatching and model identification delay will all lead to an ineluctable estimation error \cite{li2016effectiveness,fan2011impact,xiang2018impact,zhang2022fast}. In particular, various data-driven approaches have recently been proposed based on the SSM or for replacing some calculation components of the SSM for which the readers are referred to  \cite{gao2019long,Ghosh2024Dense} 
	and \cite{Pinto2023DMTT,Zhang24transformer,Zhang24L-GBM} and the references therein. 

	Instead of estimating the {discrete-time point state} of the target based on a meticulously designed SSM
	, we are interested in estimating the {continuous-time trajectory} which is given by a curve function of time (FoT) \cite{li2018joint,li2023target}. 
	This class of data-driven methods is Markov-free and aims to estimate the trajectory directly rather than the state, as illustrated on the right part of Fig. \ref{fig:Taxonomy}. 
	The T-FoT not only provides the dynamics information of the target such as the position, velocity, and acceleration, but also the high-level information such as the movement pattern of the target and smoothness of the trajectory, making it advantageous in comparing with traditional point estimators. Even more importantly, it inherently accommodates unevenly-arrived measurements and stochastic correlation between measurement noises over time \cite{Li25TFoT-part3}.  
	Relevant attempts on continuous-time trajectory modeling can be found in 
	the cutting-edge survey for continuous-time state estimation in a broader realm of robotics \cite{Talbot2024continuoustimestate}. However, most of these approaches, including our previous work \cite{li2016fitting,li2023target}, assume the deterministic real state/trajectory and do not provide uncertainty about the trajectory estimate. This identified research gap will be explored and resolved in this work and the companion papers.  

	\begin{figure}[htbp]
		\centerline{\includegraphics[width=0.9\columnwidth]{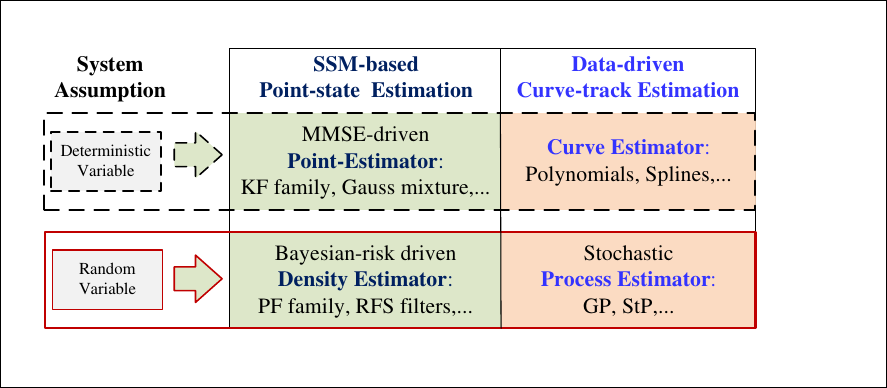}}
		\caption{A taxonomy of existing prominent estimators for target tracking}
		\label{fig:Taxonomy}
	\end{figure}
	
	\subsection{Introduction to Companion Papers}
	
	Our series of studies are founded on modeling the target trajectory using a trajectory function of time (T-FoT) \cite{li2018joint,li2023target}. 
	That is, the evolution of the target state over time is modeled by T-FoT $f:\mathbb{R}^+ \rightarrow \mathcal{X}$ in the spatio-temporal space and 
	the target state at time $t$ is given by  
	\begin{equation} \label{eq:T-FoT}
		{\mathbf{x}_t} = f(t), 
	\end{equation}
	where $t\in \mathbb{R}^+$ denotes the time and $\mathcal{X}$ denotes the state space.
	
	To take advantage of the latent state-over-time correlation and to provide an assessment of the uncertainty associated with the estimate of the T-FoT, within our series of companion papers \cite{Li25TFoT-part2,Li25TFoT-part3} including this one, we further model 
	the collection of the target states in time series as a stochastic process (SP) $\mathcal{F}\triangleq \{\mathbf{x}_t: t\in \mathbb{R}^+ \} $.  
	That is, any T-FoT is a sample path of this SP, that is,
	\begin{equation}
		f\sim \mathcal{F} .
	\end{equation}
	
	Based on the deterministic-stochastic decomposition approach rooted in Wold and Cram\'er's decomposition theorems \cite{Wold1938,Cramer1961,Box1994}, the SP can be decomposed into a deterministic FoT that captures the trend of the SP $\mathcal{F}$ and the residual SP (RSP) contains all the randomness, i.e., 
	\begin{equation}
		f(t)=F(t;\mathbf{C})+\epsilon(t) \label{eq:cramer}
	\end{equation}
	where the deterministic FoT $F(\cdot; \mathbf{C})$ is specified by parameters $\mathbf{C}$ and $\epsilon(\cdot) \sim \mathcal{E}(\cdot; \varTheta) $ denotes the RSP specified by parameters $\varTheta$. Note that $\epsilon(t)$ was interpreted as the fitting error of $F(t; \mathbf{C})$ to $f(t)$ in \cite{li2018joint,li2023target} but no SP model was built for it. 
	
	The contributions of our series of companion papers including three parts are structured as follows.  
	\begin{itemize}
		\item Part I \cite{Li25TFoT-part1} proposed a metric for evaluating the quality of any T-FoT estimate $\hat{f}$, which will be used in this paper. This metric actually provides a distance between any two trajectories given in terms of FoT. 
		\item Part II, as detailed in this paper, addresses the online fitting of a specific polynomial FoT $F(\cdot;\mathbf{C})$ subject to regularized optimization. 
		Two distinct regularization strategies are proposed to strike an optimal balance between fitting accuracy and T-FoT simplicity. 
		\item Part III \cite{Li25TFoT-part3} offers solutions for learning the RSP $\epsilon(\cdot) \sim \mathcal{E}(\cdot; \varTheta) $ (as well as the correlation between the measurement noises) for which two specific representative SPs are considered, respectively: the Gaussain process (GP) and Student's $t$ process (StP).
	\end{itemize}

	\subsection{Contribution and Organization of This Paper} 
	%
	We redefine the classic target tracking problem as an T-FoT-oriented SP learning task based on deterministic-stochastic decomposition \eqref{eq:cramer}. By this, learning the T-FoT $F(t;\mathbf{C})$ and the fitting error $e\left( t \right)$ constitutes the core of our two contributions to SP-based tracking including this paper and the companion paper \cite{Li25TFoT-part3}, respectively. 
	Our first contribution in this paper is to apply the most known polynomial trend decomposition \cite{Hodrick1997postwar,Urbin2012time}, i.e., we assume the trend T-FoT $F(t;\mathbf{C})$ by a polynomial. Although polynomial fitting in a sliding time window has been implemented earlier \cite{Rudd94,Wang94,Anderson-Sprecher96,li2018joint,li2023target}, no SP modeling and decomposition has been recognized. However, the polynomial \cite{Fan96LocalPolynomial} is quite flexible and can be analytically optimized, but faces two notable risks, namely underfitting and overfitting. To avoid this, we propose two distinctive strategies for regularization in order to select a proper polynomial order or a sparse number of polynomial parameters. It should also be noted that the RSP will usually be zero mean if the learned trend T-FoT $F(\cdot;\mathbf{C})$ is chosen as the mean function of the RSP; see the proof given in \cite{Li25TFoT-part3}. Therefore, the learning of $F(\cdot;\mathbf{C})$ plays also an important role in determining the RSP. 
	
	
	The remainder of this paper is organized as follows. Preliminary work is addressed in Section \ref{sec:Preliminaries}, including the constrained T-FoT model and $\ell_0$-regularization optimality. 
	Consequently, two distinctive approximate solvers for the optimization problem are given in Sections \ref{sec:ORLS} and \ref{sec:Newton}, respectively. 
	Extension to the multiple target trajectory case is briefly discussed in Section \ref{sec:extension-MTT}. The simulation results are given in Section \ref{sec:simulation} before the paper is concluded in Section \ref{sec:conclusion}.

	\section{Preliminaries}\label{sec:Preliminaries}

	
	

	\subsection{Polynomial T-FoT}
	By harnessing the principle of Taylor series expansion, it becomes feasible to employ higher-order polynomials to represent any 
	smooth FoT. 
	Polynomial fitting has proven effective in fitting measurements \cite{Tian22PolyFit} and has demonstrated superiority and flexibility in T-FoT modeling in our earlier works \cite{li2018joint,li2018single,li2023target}. The polynomial T-FoT that spans multiple state dimensions is simply given as follows 
	\begin{equation}\label{eq:polynomial}
		F\left ( t;\mathbf{C}_\gamma \right )= \sum_{i=0}^{\gamma }\mathbf{c}_{i}t^i,
	\end{equation}
	where $\gamma $ refers to  the order of the fitting function which may be given exactly in advance or specified with a higher bound, $\mathbf{C}_\gamma=\left\{ \mathbf{c}_i \right\} _{i=0,1,\cdots ,\gamma}$ represents the polynomial trajectory coefficients, $\mathbf{c}_i= \big[c_i^{(1)}, c_i^{(2)}, \cdots, c_i^{(r)}  \big]$, $r$ indicates the dimension in the concerning state space $\mathcal{X}$, 
	$\mathbf{c}_0,\mathbf{c}_1,\mathbf{c}_2$ correspond to the initial position, velocity, and acceleration of the target. 
	
	In many real-world applications, the trajectory is only defined in the position space. This is preferable when measurements are made on the position. 
	Moreover, by computing the derivatives of the position T-FoT with respect to time $t$,  we can estimate the velocity and acceleration of the target as follows:
	\begin{equation}
		\frac{\partial f\left( t \right)}{\partial t}\Bigg|_{t=0}  =\mathbf{c}_1, ~~~
		\frac{\partial ^2f\left( t \right)}{\partial ^2t}\Bigg|_{t=0}  =\mathbf{c}_2. \label{eq2.3}
	\end{equation}
	
	As a rule of thumb, $\gamma = 1$ and $\gamma = 2$ are sufficient to model constant velocity and constant acceleration, respectively. This validates the interpretability of the T-FoT model and an advantage in comparison with some other forms of FoT such as the B-spline \cite{Hadzagic2011batchSpline,Furgale12,Tirado2022Spline} and so on \cite{Pacholska20}.  
	
	\subsection{Sliding Time-Window T-FoT Fitting}
	In order to manage the complexity of the track function, it is common to use a time window to estimate the trajectory parameters. The default time window is given in a sliding manner as $K = [k', k]$, where $k'=\max(1,k-T_\text{w})$, $k$ denotes the current time, $T_\text{w}$ represents the supposed maximum length of the time window, and the operator $\max(a,b)$ produces the maximum between $a$ and $b$. 
	
	
	The real target T-FoT $f(t)$ 
	is measured in discrete time-series, i.e., 
	\begin{equation}\label{eq:measurement}
		\mathbf{y}_k=h_k\left( f(k),\mathbf{v}_k \right), 
	\end{equation}
	where $\mathbf{y}_k \in \mathcal{Y}$ denotes the measurement received at time $k\in \mathbb{N}^+$, $\mathcal{Y}$ denotes the measurement space, $h_k(\cdot,\cdot): \mathcal{X} \times \mathcal{V} \rightarrow \mathcal{Y}$ and $\mathbf{v}_k \in \mathcal{V}$ denote the measurement function and noise at measuring time $k$, respectively.
	
	To minimize the fitting error $\epsilon(t)$ as given in \eqref{eq:cramer}, the parameters of the fitting T-FoT should be determined as follows
	\begin{equation} \label{eq:C_k}
		\hat{\mathbf{C}}_{\gamma}=\underset{\mathbf{C}_{\gamma}}{\text{arg}\min} \mathcal{D}_K(\mathbf{C}_{\gamma}) 
		,
	\end{equation}
	where the data fitting error as adopted in this paper is defined in the least squares (LS) sense  
	\begin{equation}
		\mathcal{D}_K(\mathbf{C}_{\gamma})\triangleq \sum_{t=k'}^k{\lVert \mathbf{y}_t-h_t\left( F\left( t;\mathbf{C}_\gamma \right) ,\bar{\mathbf{v}}_t \right) \rVert^2_{\text{var}(\mathbf{y}_t)}},
	\end{equation} 
	where 
	\(\bar{\mathbf{v}}_t\) denotes the expectation of the measurement error, \(\text{var}(\mathbf{y}_t)\) denotes the variance of the measurement $\mathbf{y}_t$ and \( \left \|  \mathbf{z} \right \|^{2}_P \) denotes the normalized \(\ell_2\)-norm distance, also known as the Mahalanobis distance, as follows 
	\begin{equation}
		\|  \mathbf{z}  \|^{2}_\mathbf{P} \triangleq \mathbf{z}^\mathrm{T}\mathbf{P}^{-1}\mathbf{z}.
	\end{equation}
	It is known that the above weighted LS estimate is also the maximum likelihood estimate if $\mathbf{y}_t-h_t\left( F\left( t;\mathbf{C}_\gamma \right) ,\bar{\mathbf{v}}_t \right)$ follows a zero-mean Gaussian distribution. 
	
	
	To be more general, in the presence of any a-priori model information or system constraint, it may be incorporated into
	the optimization cost function by adding a
	regularization factor $\varOmega _F(\mathbf{C}_{\gamma})$ as a measure of the disagreement of the fitting function with the a-priori model constraint, leading to
	\begin{equation} \label{eq:Const-T-FoT-Op}
		\mathrm{\bf Problem ~ 1}: ~~ \hat{\mathbf{C}}_{\gamma}=\underset{\mathbf{C}_{\gamma}}{\text{arg}\min}\left( \mathcal{D}_K(\mathbf{C}_{\gamma})+\lambda \varOmega _F\left( \mathbf{C}_{\gamma} \right) \right),
	\end{equation}
	where the coefficient $\lambda$ trades off between the data fitting error and the model/constraint fitting error. 
	
	In this paper, we will consider two regularization functions $\varOmega _F\left( \mathbf{C}_{\gamma} \right) \triangleq \gamma +1 $ and $\varOmega _F\left( \mathbf{C}_{\gamma} \right) \triangleq  \rVert \mathbf{C}_{\gamma}  \rVert _0 $, respectively, where  $\lVert \mathbf{C} \rVert _0$ denotes the $\ell_0$ norm of $\mathbf{C}$, counting the number of non-zero elements of $\mathbf{C}$. Two optimization solvers are proposed in the following Sections \ref{sec:ORLS} and \ref{sec:Newton}, respectively. One employs a grid-searching approach based on an upper bound of the optimal order of the polynomial while the other employs the hybrid Newton method.

	\subsection{$\ell_0$-regularization Optimality}
	There are some useful definitions and results on the smoothness, convexity and optimality of the T-FoT. 
	Consider the following optimization,
	\begin{equation}\label{eq:l0_form}
		\underset{\mathbf{C}\in \mathbb{R}^n}{\min}\ g\left( \mathbf{C} \right) +\lambda \lVert \mathbf{C} \rVert _0.
	\end{equation}

	\begin{definition}[L-Smooth]
		Let $L \ge 0$. A function $g$ is said to be L-smooth if there exists $\bigtriangledown g$ (differentiable) and the following inequality holds for all $\mathbf{u},\mathbf{w}\in \mathbb{R}^n$ \cite{beck2017first}, 
		\begin{equation}\label{eq:smooth}
			\lVert \bigtriangledown g\left( \mathbf{u} \right) -\bigtriangledown g\left( \mathbf{w} \right) \rVert _*\le L\lVert \mathbf{u}-\mathbf{w} \rVert,
		\end{equation}
		where $\lVert \cdot \rVert$ and $\lVert \cdot \rVert _*$ are a pair of dual norms, where $\lVert \mathbf{u} \rVert _*=sup\left\{ \mathbf{u}^\text{T} \mathbf{v}|\lVert \mathbf{v} \rVert \le 1 \right\}$, the constant $L$ is called the smoothness parameter.
	\end{definition}
	
	The first useful result on L-smooth functions is the descent lemma, which states that they can be upper bounded by a certain quadratic function. 
	\begin{lemma}[{descent lemma}] Let $g$ be an L-smooth function ($L \ge 0$). Then  for any $\mathbf{u},\mathbf{w}\in \mathbb{R}^n$,
		\begin{equation}
			g\left( \mathbf{w} \right) \le g\left( \mathbf{u} \right) +\left< \bigtriangledown g\left( \mathbf{u} \right) ,\mathbf{w}-\mathbf{u} \right> +\left(L/2\right)\lVert \mathbf{w}-\mathbf{u} \rVert ^2,
		\end{equation}
		where $\left<\cdot ,\cdot \right> $ represents the inner product. 
	\end{lemma}
	The proof can be found in \cite[Lemma 5.7]{beck2017first}.
	
	\begin{definition}[$\ell$-strongly convex]
		A function $g$ is called $\ell$-strongly convex for a given $\ell>0$ if {\rm dom}$(g)$ is convex and the following inequality holds for any $\mathbf{u},\mathbf{w}\in$ {\rm dom}$(g)$ and $\varphi \in \left [ 0,1 \right ] $,
		\begin{align}
			g( \varphi \mathbf{u} +\left(1-\varphi \right)\mathbf{w} ) \le  &  \varphi g\left( \mathbf{u} \right) +\left(1-\varphi \right)g\left( \mathbf{w} \right) \nonumber \\
			& -\frac{\ell}{2}\varphi\left(1-\varphi \right)\lVert  \mathbf{u}-\mathbf{w} \rVert ^2.
		\end{align}
	\end{definition}
	
	\begin{lemma}
		If there exists $\bigtriangledown g$ (differentiable), the function $g$ is strongly convex with a constant $\ell>0$  if
		\begin{equation}\label{eq:convex}
			g\left( \mathbf{w} \right) \ge g\left( \mathbf{u} \right) +\left< \bigtriangledown g\left( \mathbf{u} \right) ,\mathbf{w}-\mathbf{u} \right> +\left(\ell/2\right)\lVert \mathbf{w}-\mathbf{u} \rVert ^2.
		\end{equation}
	\end{lemma}
	The proof can be found in \cite[Theorem 5.24]{beck2017first}.
	
	\begin{definition}[$\tau$-stationary]
		A point $\mathbf{C}\in \mathbb{R}^n$ is called a $\tau$-stationary point  \cite{zhou2021newton} of the problem \eqref{eq:l0_form} if there is a $\tau >0$ satisfying
		\begin{equation}\label{eq:tau_stati_point}
			\begin{aligned}
				\mathbf{C} &\in \text{Prox}_{\tau \lambda \lVert \cdot \rVert _0}\left( \mathbf{C}-\tau \bigtriangledown g\left( \mathbf{C} \right) \right) \\
				& \triangleq \underset{\mathbf{D}\in \mathbb{R}^n}{\text{arg}\min}\frac{1}{2}\lVert \mathbf{D}-\left( \mathbf{C}-\tau \bigtriangledown g\left( \mathbf{C} \right) \right) \rVert ^2+\tau \lambda \lVert \mathbf{D} \rVert _0.
			\end{aligned}
		\end{equation}
	\end{definition}
	
	Consider now the proximal operator $\text{Prox}_{\tau \lambda \lVert \cdot \rVert _0}$ \cite{attouch2013convergence} 
	expressed in a closed form as follows
	\begin{equation}
		\begin{aligned}
			&\text{Prox}_{\tau \lambda \lVert \cdot \rVert _0}\left( \mathbf{C} \right)=  \\
			&\left\{\text{Prox}_{\tau \lambda \lVert \cdot \rVert _0}\left( \mathbf{c}_0 \right),\text{Prox}_{\tau \lambda \lVert \cdot \rVert _0}\left(\mathbf{c}_1 \right),\cdots,\text{Prox}_{\tau \lambda \lVert \cdot \rVert _0}\left( \mathbf{c}_\gamma \right)\right\},
		\end{aligned}
	\end{equation}
	where
	\begin{equation}\label{eq:prox}
		\text{Prox}_{\tau \lambda \lVert \cdot \rVert _0}\left( \mathbf{c}_i \right)  = \begin{cases}
			0,&\left| \mathbf{c}_i \right|<\sqrt{2\tau \lambda}\\
			0\ or\ \mathbf{c}_i,&\left| \mathbf{c}_i \right|=\sqrt{2\tau \lambda} \\
			\mathbf{c}_i,& \left| \mathbf{c}_i \right|>\sqrt{2\tau \lambda}
		\end{cases}
	\end{equation}
	
	The stationarity is a necessary condition for the local optimality of problem \eqref{eq:l0_form}. In case of convex $g$, it is a necessary and sufficient global optimality condition \cite{beck2018proximal}.
	
%
%
%
	
	\section{Polynomial T-FoT with Limiting Order} 
	\label{sec:ORLS}
	
	Polynomial order $\gamma$ is a critical parameter in trajectory fitting \cite{stoica2004model}, and its value may exhibit time-varying characteristics in response to the maneuvering behavior of the target.
	The improper fitting model faces two notorious challenges, a.k.a. underfitting and overfitting due to underestimated and overestimated polynomial order, respectively. As illustrated in Figs. \ref{fig:diff_order_poly} (a) and (c), the former manifests itself when a polynomial of insufficient order is selected 
	while the latter typically transpires when a polynomial of excessively high order is chosen. 

	\begin{figure}[htbp]
		\centerline{\includegraphics[width=0.9\columnwidth,trim=100 2 90 20,clip]{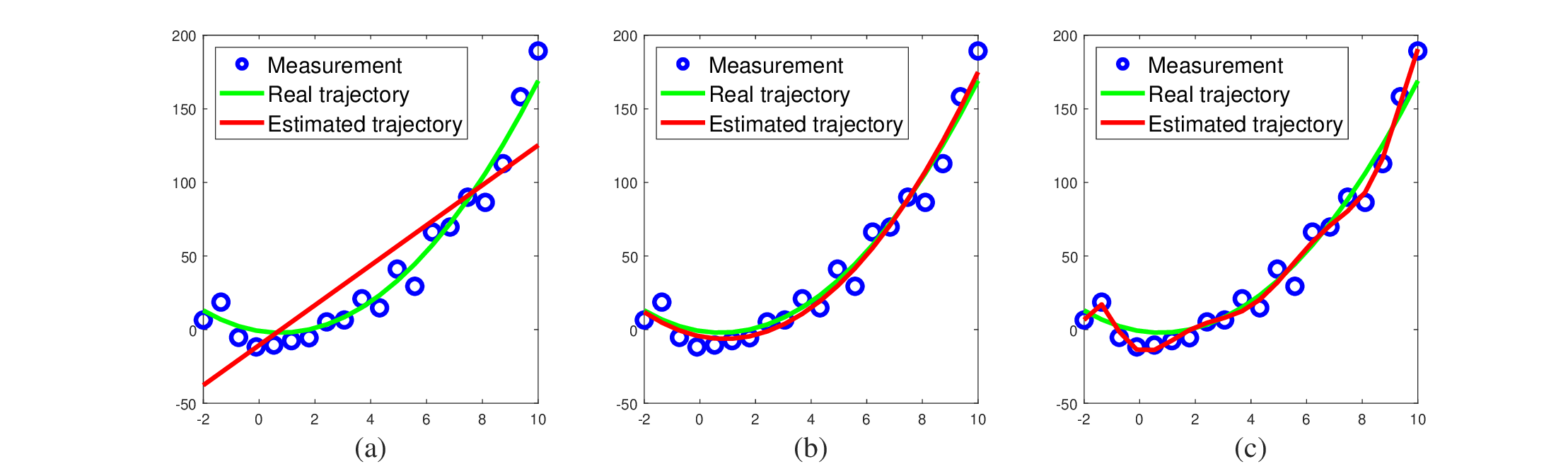}}
		\caption{Fitting by polynomials of different orders.}
		\label{fig:diff_order_poly}
		\vspace{-2mm}
	\end{figure}
	
	Our first proposed approach defines the regularization item as $\varOmega _F\left( \mathbf{C}_{\gamma} \right) \triangleq \gamma +1 $. Then, Problem 1 is reduced to
	\begin{equation}\label{eq:Ck_order}
		\mathrm{\bf Problem ~ 2}: ~~ \hat{\mathbf{C}}_{\gamma}=\underset{\mathbf{C}_{\gamma}}{\text{arg}\min}\left( \mathcal{D}_K(\mathbf{C}_{\gamma}) +\lambda \left( {\gamma + 1} \right) \right).
	\end{equation}
	
	In Problem 2, the fitting error of the polynomial gradually decreases as the order increases, while 
	the regularization term increases by \(\lambda\) with each increment in the order as illustrated in Fig. \ref{fig:order_error}. After a certain point where the lowest overall cost exists, the constant increase of the order penalty will be more significant than the decrease of the data fitting error. 
	As such, we can start 
	with $\gamma =0$ and incrementally increase the order while monitoring the increase in the regularization term, until the lowest overall cost was found at $\gamma^*$. This grid search allows us to efficiently obtain the optimal order $\gamma^*$ in a narrow, bounded range. 
	Therefore, we set \textit{halting condition} for the grid search as if either $\gamma$ reaches the maximum order (cf. Eq. \eqref{eq:gamma_UpBound}) or if the reduction of the fitting error is less than the increase of the order cost at each step, namely,
	\begin{equation} \label{eq:Halt_con}
		|\mathcal{D}_K(\hat{\mathbf{C}}_{\gamma^*}) -\mathcal{D}_K(\hat{\mathbf{C}}_{\gamma^*+1}) | \le \lambda,
	\end{equation}
	where $\mathcal{D}_K(\hat{\mathbf{C}}_{\gamma^*+1})$ denotes the data fitting error corresponding to a polynomial order of $\gamma^*+1$.

	\begin{figure}[htbp]
		\centerline{\includegraphics[width=0.7\columnwidth]{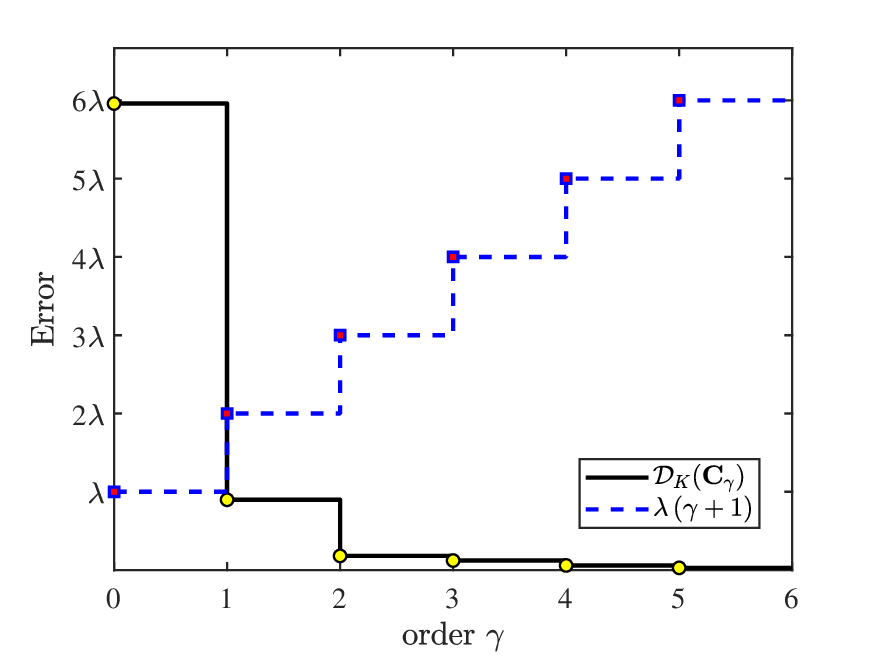}}
		\caption{Illustration of the monotonous decrease of data fitting error $\mathcal{D}_K(\mathbf{C}_{\gamma})$ and the monotonous increase of order cost with the increase of $\gamma$. 
		}
		\label{fig:order_error}
	\end{figure}

	\subsection{Bounds for $\gamma^ *$ and $\lambda$}
	
	\subsubsection{Bounds of \({\gamma^ * }\)}
	We propose to narrow the range of the optimal choice \({\gamma^ * }\) by imposing the following upper bound.

	\begin{proposition}
		In accordance with the principle that the error associated with the optimal order is less than that of other orders, an upper bound for the optimal order $ \gamma^{*}$ is given
		\begin{equation}\label{eq:gamma_UpBound}
			\gamma^{*} \le \frac{\mathcal{D}_K(\hat{\mathbf{C}}_1)}{\lambda } + 1.
		\end{equation}
	\end{proposition}
	
	\begin{proof}
		The result is rooted in the fact that the overall fitting error in the case of the optimal value of \({\gamma^ * }\) is no greater than that of the first-order, that is 
		\begin{equation}
			\mathcal{D}_K(\hat{\mathbf{C}}_1)+2\lambda \geq  \mathcal{D}_K(\hat{\mathbf{C}}_{\gamma^{*}})+(\gamma^{*}+1)\lambda.
		\end{equation}
		This together with the fact $\mathcal{D}_K(\hat{\mathbf{C}}_{\gamma^{*}}) \ge 0$ leads to \eqref{eq:gamma_UpBound}.
	\end{proof}

	
	
	\subsubsection{Bounds of \( \lambda \)}
	The regularization parameter \( \lambda \), which encourages sparsity in the polynomial coefficients, plays a vital role in balancing the fitting capacity and the sparsity of the model. 
	In other words, a larger \( \lambda \) exerts a stronger regularization and consequently more parameters are driven toward zero, resulting in a sparser model. 
	Note that model sparsity is not equivalent to a lower polynomial degree, as these two concepts represent distinct aspects of model complexity.

	\begin{proposition}
		The upper and lower bounds for the factor $ \lambda$ are given
		\begin{equation} \label{eq:lambda-range}
			\frac{\mathcal{D}_K(\hat{\mathbf{C}}_1)}{T_\text{w}-2}<\lambda\le \mathcal{D}_K(\hat{\mathbf{C}}_1),
		\end{equation}
		where $T_w$ is the concerning data size. 
	\end{proposition}
	
	\begin{proof}
		Firstly, it emphasizes that \( \lambda \) should be chosen to be greater than zero, ensuring the effectiveness of the regularization term in promoting sparsity. Secondly, the selection of \( \lambda \) should satisfy the condition that it is smaller than the first-order error $\mathcal{D}_K(\hat{\mathbf{C}}_1)$, 
		that is, $ \lambda  \le \mathcal{D}_K(\hat{\mathbf{C}}_1)$.  
		
		Furthermore, 
		overfitting occurs when the polynomial order is set equal to or larger than $T_w-1$ when the polynomial passes through all the measurement points. Therefore, to avoid this, it is required that
		\begin{equation}
			\frac{\mathcal{D}_K(\hat{\mathbf{C}}_1)}{\lambda}+1<T_\text{w}-1.
		\end{equation}
		which results in the lower bound $\lambda >\frac{\mathcal{D}_K(\hat{\mathbf{C}}_1)}{T_\text{w}-2}$.
	\end{proof}


	\subsection{Order-Recursive Least Squares (ORLS) Solver for Linear Measurement} 
	
	We write the vector form of the T-FoT with regard to all time-instants in the time-window $K=[k',k]$ as follows, 
	\begin{equation} \label{eq:linearized-measurement}
		F\left(K;\mathbf{C}_{\gamma} \right)  = \mathbf{Z}_{\gamma}\mathbf{C}_{\gamma},
	\end{equation}
	where $\mathbf{Z}_{\gamma}$ denotes a Vandermonde matrix of size $T_\text{w}\times (\gamma+1)$, 
	structured as follows,
	\begin{equation}
		\mathbf{Z}_{\gamma} \triangleq \left[ \begin{matrix}
			1&		{k'}&		\cdots&		{k'}^{\gamma}\\
			1&		{k'+1}&		\cdots&		{(k'+1)}^{\gamma}\\
			\vdots&		\vdots&		\ddots&		\vdots\\
			1&		k&		\cdots&		{k}^{\gamma}\\
		\end{matrix} \right].
	\end{equation}
	Then, given linear measurement with additive, zero-mean noise, namely $\mathbf{y}_k= \mathbf{x}_k + \mathbf{v}_k, \bar{\mathbf{v}}_t =\mathbf{0}$, we get 
	\begin{align} \label{eq:D_k(C)-Linear}
		\mathcal{D}_K(\mathbf{C}_{\gamma})&=\lVert \mathbf{Y}_K -\mathbf{Z}_{\gamma}\mathbf{C}_{\gamma} \rVert^2_{\text{var}(\mathbf{Y})} \nonumber \\
		&=\left(\mathbf{Y}_K -\mathbf{Z}_{\gamma}\mathbf{C}_{\gamma}\right)^\text{T}{\text{var}(\mathbf{Y})}^{-1}\left(\mathbf{Y}_K -\mathbf{Z}_{\gamma}C_{\gamma}\right),
	\end{align}
	where the  position measurements is denoted as
	\begin{equation}\label{eq:PositionMeasuremSet}
		\begin{aligned}
			\mathbf{Y}_K \triangleq &\left[ \begin{matrix}
				\mathbf{y}_{k'},&		\mathbf{y}_{k'+1},&		\cdots,&		\mathbf{y}_k\\
			\end{matrix} \right] ^{\text{T}}.
		\end{aligned}
	\end{equation}

	In minimizing $\mathcal{D}_K(\mathbf{C}_{\gamma})$, if the uncertainty of the observed data is time-invariant (i.e., the noise is homogeneous), the term $\text{var}(\mathbf{y}_t)$ can be disregarded, resulting in a standard LS fitting. Otherwise, we can factorize it as ${\text{var}(\mathbf{Y})} ^{-1}=\mathbf{A}^{T}\mathbf{A}$ with positive definite $T_\text{w} \times T_\text{w} $ matrix $\mathbf{A}$ according to the Cholesky factorization method \cite{Ye21Cholesky}. 
	This yields 
	\begin{align}
		\mathcal{D}_K(\mathbf{C}_{\gamma})&=\left(\mathbf{Y}_K -\mathbf{Z}_{\gamma}\mathbf{C}_{\gamma}\right)^\text{T}\mathbf{A}^{T}\mathbf{A}\left(\mathbf{Y}_K -\mathbf{Z}_{\gamma}\mathbf{C}_{\gamma}\right) \nonumber \\
		&=\left(\tilde{\mathbf{Y}}_K -\tilde{\mathbf{Z}}_{\gamma}\mathbf{C}_{\gamma}\right)^\text{T}\left(\tilde{\mathbf{Y}}_K -\tilde{\mathbf{Z}}_{\gamma}\mathbf{C}_{\gamma}\right),
	\end{align}
	where $\tilde{\mathbf{Y}}_K \triangleq \mathbf{A}\mathbf{Y}_K$ and $\tilde{\mathbf{Z}}_{\gamma} \triangleq \mathbf{A}\mathbf{Z}_{\gamma}$. 
	
	So far, $\mathcal{D}_K(\mathbf{C}_{\gamma})$ is now formulated as conventional LS problem 
	for which the solution is 
	\begin{equation}\label{eq:hat(C)-gamma}
		\hat{\mathbf{C}}_{\gamma}  = {\left( {{\tilde{\mathbf{Z}}_{\gamma}^\text{T}}\tilde{\mathbf{Z}}_{\gamma}} \right)^{ - 1}}{\tilde{\mathbf{Z}}_{\gamma}^\text{T}}\tilde{\mathbf{Y}}_K .
	\end{equation}
	\begin{equation}
		\mathcal{D}_K(\hat{\mathbf{C}}_{\gamma})=\left(\tilde{\mathbf{Y}}_K -\tilde{\mathbf{Z}}_{\gamma}\hat{\mathbf{C}}_{\gamma}\right)^\text{T}\left(\tilde{\mathbf{Y}}_K -\tilde{\mathbf{Z}}_{\gamma}\hat{\mathbf{C}}_{\gamma}\right).
	\end{equation}
	
	We then employ the ORLS algorithm \cite[Sec. 8]{steven1993fundamentals} to iteratively update the parameters to the optimal one. 
	Each time, we increase the order by one, adding a column to the time matrix, resulting in a new matrix that can be partitioned as follows
	\begin{equation}
		\tilde{\mathbf{Z}}_{\gamma +1}=\left[ \begin{matrix}
			\tilde{\mathbf{Z}}_{\gamma}&		\tilde{\mathbf{z}}_{\gamma +1}\\
		\end{matrix} \right] ,
	\end{equation}
	where $\tilde{\mathbf{z}}_{\gamma +1}=\mathbf{A}\mathbf{z}_{\gamma +1}$, $\mathbf{z}_{\gamma +1}=\left[k'^{\gamma+1}, (k'+1)^{\gamma+1}, \cdots, k^{\gamma+1} \right]^\text{T}$.
	
	To update $\hat{\mathbf{C}}_{\gamma}$ and $\mathcal{D}_K(\hat{\mathbf{C}}_{\gamma})$ we use \cite[Eq.(8.28)]{steven1993fundamentals}
	\begin{equation}\label{eq:c_recurisive}
		\hat{\mathbf{C}}_{\gamma +1}=\begin{bmatrix}
			\hat{\mathbf{C}}_{\gamma}-\frac{\left( \tilde{\mathbf{Z}}_{\gamma}^{T}\tilde{\mathbf{Z}}_{\gamma} \right) ^{-1}\tilde{\mathbf{Z}}_{\gamma}^{T}\tilde{\mathbf{z}}_{\gamma +1}\tilde{\mathbf{z}}_{\gamma +1}^{T}\mathbf{W}_{\gamma}^{\bot}\tilde{\mathbf{Y}}_K }{\tilde{\mathbf{z}}_{\gamma +1}^{T}\mathbf{W}_{\gamma}^{\bot}\tilde{\mathbf{z}}_{\gamma +1}}\\
			\frac{\tilde{\mathbf{z}}_{\gamma +1}^{T}\mathbf{W}_{\gamma}^{\bot}\tilde{\mathbf{Y}}_K }{\tilde{\mathbf{z}}_{\gamma +1}^{T}\mathbf{W}_{\gamma}^{\bot}\tilde{\mathbf{z}}_{\gamma +1}}
		\end{bmatrix},
	\end{equation}
	where 
	\begin{equation}
		\mathbf{W}_\gamma^{\bot} \triangleq \mathbf{I} - {\tilde{\mathbf{Z}}_\gamma}\mathbf{B}_{\gamma}\tilde{\mathbf{Z}}_\gamma^\text{T}.
	\end{equation}
	where $\mathbf{I}$ denotes the identity matrix. $\mathbf{W}_\gamma^{\bot}$ is the projection matrix onto the subspace orthogonal to that spanned by the columns of $\tilde{\mathbf{Z}}_{\gamma}$ and the abbreviation
	\begin{equation}\label{eq4.7}
		\mathbf{B}_{\gamma} \triangleq \left( \tilde{\mathbf{Z}}_{\gamma}^{T}\tilde{\mathbf{Z}}_{\gamma} \right) ^{-1},
	\end{equation}
	is updated recursively as follows \cite[Eq.(8.30)]{steven1993fundamentals}
	\begin{equation}\label{eq:D_recurisive}
		{\mathbf{B}_{\gamma + 1}} = \begin{bmatrix}
			{{\mathbf{B}_\gamma} + \frac{{{\mathbf{B}_\gamma}\tilde{\mathbf{Z}}_\gamma^\text{T}{\tilde{\mathbf{z}}_{\gamma + 1}}\tilde{\mathbf{z}}_{\gamma + 1}^\text{T}{\tilde{\mathbf{Z}}_\gamma}{\mathbf{B}_\gamma}}}{{\tilde{\mathbf{z}}_{\gamma + 1}^\text{T}{\mathbf{W}_\gamma^{\bot}}{\tilde{\mathbf{z}}_{\gamma + 1}}}}}&{ - \frac{{{\mathbf{B}_\gamma}\tilde{\mathbf{Z}}_\gamma^\text{T}{\tilde{\mathbf{z}}_{\gamma + 1}}}}{{\tilde{\mathbf{z}}_{\gamma + 1}^\text{T}{\mathbf{W}_\gamma^{\bot}}{\tilde{\mathbf{z}}_{\gamma + 1}}}}}\\
			{ - \frac{{\tilde{\mathbf{z}}_{\gamma + 1}^\text{T}{\tilde{\mathbf{Z}}_\gamma}{\mathbf{B}_\gamma}}}{{\tilde{\mathbf{z}}_{\gamma + 1}^\text{T}{\mathbf{W}_\gamma^{\bot}}{\tilde{\mathbf{z}}_{\gamma + 1}}}}}&{\frac{1}{{\tilde{\mathbf{z}}_{\gamma + 1}^\text{T}{\mathbf{W}_\gamma^{\bot}}{\tilde{\mathbf{z}}_{\gamma + 1}}}}}
		\end{bmatrix}.
	\end{equation}
	The LS fitting error is reduced by \cite[Eq.(8.31)]{steven1993fundamentals}
	\begin{equation}\label{eq:MIN_recurisive}
		\mathcal{D}_K(\hat{\mathbf{C}}_{\gamma})-\mathcal{D}_K(\hat{\mathbf{C}}_{\gamma+1}) = \frac{{{{\left( {\tilde{\mathbf{z}}_{\gamma + 1}^\text{T}{\mathbf{W}_\gamma^{\bot}}\tilde{\mathbf{Y}}_K } \right)}^2}}}{{\tilde{\mathbf{z}}_{\gamma + 1}^\text{T}{\mathbf{W}_\gamma^{\bot}}{\tilde{\mathbf{z}}_{\gamma + 1}}}}.
	\end{equation}
	
	

	Employing ORLS avoids the computationally expensive matrix inversion at each iteration. The iteration to find the optimal $\lambda^*$ will be terminated with the monitoring of the error reduction as given in \eqref{eq:MIN_recurisive}, according to the halting condition \eqref{eq:Halt_con}. We summarize the process of the algorithm in Algorithm. \ref{alg_ORLS}.
	
	\begin{algorithm}[!ht]
		\caption{The $\gamma$-limiting ORLS algorithm}\label{alg_ORLS}
		\KwIn{Sensor measurement \(\mathbf{Y}_K\) in the time-window $K$, T-FoT parameters: \( \lambda \), $\gamma=0$.}
		\KwOut{Optimal polynomial order \(\gamma^{*}\), polynomial coefficient \(\hat {\mathbf{C}}_{\gamma^{*}}\).}
		\begin{algorithmic}[1]
			\STATE {Initialize \(\hat {\mathbf{C}}_{0}\), \(\mathcal{D}_K (\hat{\mathbf{C}}_0) \) and \(\mathbf{B}_0\);}
			\WHILE {the halting conditions are not violated}
			{
				\STATE {\(\gamma \leftarrow \gamma+1\);}
				\STATE {calculate \eqref{eq:hat(C)-gamma} and \eqref{eq:c_recurisive} for updating parameters;}
				\STATE {calculate 
					\eqref{eq:MIN_recurisive} for checking the halting condition;}
			}
			\ENDWHILE
			\STATE {\textbf{return} \(\gamma^{*}=\gamma, \hat{\mathbf{C}}_{\gamma^{*}} = \hat{\mathbf{C}}_{\gamma}\) ;}
		\end{algorithmic}
	\end{algorithm}

	\subsection{Extension to Nonlinear Measurement} \label{sec:nonlinearMeasurement}
	Nonlinear measurements generally do not admit a closed-form LS solution for the T-FoT optimization problem. Instead, the corresponding nonlinear LS optimization must be based on computationally intensive iterative/numerical solvers. Two alternatives that can mitigate this challenge 
	are considerable. One converts the non-linear measurement to the linear position measurements as discussed in \cite[Sec. 3]{li2023target}. This is also referred to as the transformation of parameters \cite[Ch.8.9]{steven1993fundamentals}. This, however, applies only to the determined or over-determined measurement system, i.e., the dimension of the measurement is no lower than the dimension of the state space. 
	For example, the range-bearing measurement can be converted to position measurement for position T-FoT fitting \cite{li2023target}. 
	The other choice is to linearize the nonlinear measurement, 
	similar with what has been done within the prevalent extended Kalman filter (KF) in extending the KF. Reconsider the measurement function \eqref{eq:measurement} linearized as follows: 
	\begin{equation}
		\mathbf{y}_k=\hat{\mathbf{y}}_{k-1}+\mathbf{J}_k\left(\mathbf{x}_k-\hat{\mathbf{x}}_{k-1}\right)+\mathbf{v}_k,
	\end{equation}
	where $\hat{\mathbf{y}}_{k-1}=h(\hat{\mathbf{x}}_{k-1}, \bar{\mathbf{v}_{k+1}})$, $\hat{\mathbf{x}}_{k-1}$ denotes the state estimate extracted from the T-FoT by $\hat{\mathbf{x}}_{k-1} = F(k-1; \mathbf{C}_{k-1})$ based on the parameters $\mathbf{C}_{k-1}$ obtained at time $k-1$, and $\mathbf{J}_k$ denotes the Jacobian matrix of the partial derivatives of the measurement function $h_k(\cdot)$. 
	
	Then, given $\hat{\mathbf{x}}_{t},\hat{\mathbf{y}}_{t}$ estimated at time $t=k',k'+1,...,k-1$ and all of the measurements in the time-window $\mathbf{y}_k, k\in K$, we define  
	$ \mathbf{\delta Y}_K \triangleq [
	\mathbf{y}_{k'}-\hat{\mathbf{y}}_{k'-1} + \mathbf{J}_{k'}\hat{\mathbf{x}}_{k'-1}, \mathbf{y}_{k'+1}-\hat{\mathbf{y}}_{k'} + \mathbf{J}_{k'+1}\hat{\mathbf{x}}_{k'}, 	\cdots,		\mathbf{y}_{k}-\hat{\mathbf{y}}_{k-1} + \mathbf{J}_{k}\hat{\mathbf{x}}_{k-1} 
	] ^{\text{T}}$
	and reformulate \eqref{eq:D_k(C)-Linear} as 
	\begin{align} \label{eq:D_k(C)-ExLinear}
		\mathcal{D}_K(\mathbf{C}_{\gamma})&=\lVert  \mathbf{\delta Y}_K - \mathbf{J}_K  \mathbf{Z}_{\gamma}\mathbf{C}_{\gamma} \rVert^2_{\text{var}( \mathbf{\delta Y})},
	\end{align}
	where $\mathbf{J}_K \triangleq \text{diag} (\mathbf{J}_{k'}, \mathbf{J}_{k'+1},\cdots, \mathbf{J}_{k})$ is a diagonal matrix with diagonal elements $\mathbf{J}_{k'}, \mathbf{J}_{k'+1},\cdots, \mathbf{J}_{k}$ in order. 
	
	Comparing \eqref{eq:D_k(C)-Linear} with \eqref{eq:D_k(C)-ExLinear}, it can be seen that $\mathbf{Y}_K$ and $\mathbf{Z}_{\gamma}$ were replaced by $\delta \mathbf{Y}_K$ and $\mathbf{J}_K  \mathbf{Z}_{\gamma}$ in the latter, respectively. The $\gamma$-limiting ORLS algorithm can be similarly used to estimate $\mathbf{C}_{\gamma}$. The detail is omitted.   
	
	\section{$\ell_0$-regularized T-FoT Fitting}\label{sec:Newton}
	

	For a judicious equilibrium  between promotion of sparsity and fitting accuracy in the T-FoT model \eqref{eq:Const-T-FoT-Op},
	a prevalent approach is to enforce sparsity, typically characterized by the $\ell_0$ norm \cite{zhou2021newton} regularization. In other words, 
	Problem 1 is specified using the $\ell_0$-norm in accordance with the law of parsimony, leading to a new form: 
	\begin{equation}\label{eq:C_k_l0}
		\mathrm{\bf Problem ~ 3}: ~~ \hat{\mathbf{C}}_{\gamma}=\underset{\mathbf{C}_{\gamma}}{\text{arg}\min}\left( \mathcal{D}_K(\mathbf{C}_{\gamma})+\lambda \lVert \mathbf{C}_{\gamma} \rVert _0 \right).
	\end{equation}
	
	
	This $\ell_0$-regularized optimization problem is nonconvex, noncontinuous, and NP-hard, 
	no matter what $\mathcal{D}_K(\mathbf{C}_{\gamma})$. 
	Existing approaches 
	can be broadly categorized into two main groups: model transformation and direct processing \cite{some2020zhao}. 
	The model transformation method focuses on converting the nonconvex and nonsmooth $\ell_0$ norm 
	into a more tractable form, either by transforming it into a convex function like the $\ell_1$-regularization \cite{kim2007interior} 
	or an easy-handling nonconvex function 
	like the $\ell_{1/2}$-regularization \cite{xu2012lell} 
	and $\ell_p$ norm relaxation, where $0 < p < 1$ \cite{fung2011equivalence}. 
	These advancements have led to the emergence of various highly efficient first-order algorithms, including the iterative shrinkage-thresholding algorithm \cite{blumensath2008iterative}, 
	augmented Lagrangian method, and alternating direction multiplier method (ADMM) \cite{he2020optimal}. 
	
	
	
	The direct processing approach to $\ell_0$ optimization aims to bypass the need for exact recovery verification and broadens the applicability of the findings. 
	Representative solutions include the iterative hard-thresholding algorithm \cite{blumensath2009iterativecompressed}, active set Barzilar-Borwein algorithm \cite{cheng2020active}, 
	variational approach \cite{ito2013variational}, to name a few. 
	In what follows, we detail a direct processing solver based on the hybrid Newton approach to obtain the optimal solution of \eqref{eq:C_k_l0}.  For simplicity, we will omit the subscript of $\gamma$ (which can be set to any positive integer).
	
	For the problem \eqref{eq:C_k_l0}, a $\tau$-stationary point is  defined when a $\tau > 0$ exists that satisfies
	\begin{equation}\label{eq:eq19_stati_point}
		\mathbf{C} \in \text{Prox}_{\tau \lambda \lVert \cdot \rVert _0}\left( \mathbf{C}-\tau \bigtriangledown \mathcal{D}_K\left( \mathbf{C} \right) \right).
	\end{equation}
	
	The hybrid Newton method \cite{boyd2004convex} needs the strong smoothness and convexity of $\mathcal{D}_K(\mathbf{C})$ and aims to establish the relationship between the $\tau$-stationary point and the local/global minimizer of \eqref{eq:C_k_l0}. Next, we introduce   the following theorem, which clarifies the relationship between the $\tau$-stationary point and the local/global minimizer of \eqref{eq:C_k_l0}. Theorem \ref{theorem1} and the proof are detailed in \cite[Theorem 1]{zhou2021newton}.

	\begin{theorem} 
		\label{theorem1}
		For problem \eqref{eq:C_k_l0}, the following results hold.
		
		1) (\textbf{Necessity}) A global minimizer $\mathbf{C}^{\ast}$ is also a $\tau$-stationary point for any $0 <\tau<
		1/L$ if $\mathcal{D}_K(\mathbf{C})$ is strongly smooth with $L > 0$. Moreover,
		\begin{equation}\label{eq:necessity_global}
			\mathbf{C}^{\ast}=\text{Prox}_{\tau \lambda \lVert \cdot \rVert _0}\left( \mathbf{C}^{\ast}-\tau \bigtriangledown \mathcal{D}_K\left( \mathbf{C}^{\ast} \right) \right).
		\end{equation}
		
		2) (\textbf{Sufficiency}) A $\tau$-stationary point with $\tau>0$ is a local minimizer if $\mathcal{D}_K(\mathbf{C})$ is convex.
		Furthermore, a $\tau$-stationary point with $\tau  \ge 1/\ell$ is also a (unique) global minimizer if $\mathcal{D}_K(\mathbf{C})$ is strongly convex with $\ell>0$.
	\end{theorem}
	
	To express the solution of (\ref{eq:eq19_stati_point}) more explicitly, we define
	\begin{equation} \nonumber 
		T\triangleq T_{\tau}\left( \mathbf{C},\lambda \right) :=\left\{ i\in \mathbb{N}:\left| \mathbf{c}_i-\tau \bigtriangledown _i \mathcal{D}_K\left( \mathbf{C}\right) \right|\ge \sqrt{2\tau \lambda} \right\},
	\end{equation}
	where $\mathbb{N}=\left\{0,1,\cdots,\gamma \right\}$, $\bigtriangledown _i \mathcal{D}_K\left( \mathbf{C}\right)=\left(\bigtriangledown \mathcal{D}_K\left( \mathbf{C}\right)\right)_i$.
	
	Based on this set, we introduce the following stationary equation
	\begin{equation}\label{eq:stationary}
		G_{\tau}\left( \mathbf{C};T \right) \triangleq \left[ \begin{array}{c}
			\bigtriangledown_T \mathcal{D}_K\left( \mathbf{C}\right)\\
			\mathbf{C}_{\overline{T}}\\
		\end{array} \right] =\mathbf{0},
	\end{equation}
	where $\overline{T}$ is the complementary set of $T \subseteq \mathbb{N}$.

	
	The relationship between \eqref{eq:eq19_stati_point} and \eqref{eq:stationary} is revealed in \cite{zhou2021newton}.
	To solve \eqref{eq:stationary}, we first locate the hidden index set $T$ by employing an adaptive updating rule as follows. First, calculate an approximation $T_{\kappa}$ for a computed point $\mathbf{C}^{(\kappa)}$ (the $\kappa$-th iteration). Then, by fixing this set $T_{\kappa}$, apply the Newton method to $G_{\tau} (\mathbf{C}; T_{\kappa})$ once obtaining a direction $s^{(\kappa)}$. In other words, $s^{(\kappa)}$ is a solution to the following equation. 
	\begin{equation}
		\bigtriangledown G_{\tau}\left(\mathbf{C}^{(\kappa)};T_{\kappa} \right) s+G_{\tau}\left( \mathbf{C}^{(\kappa)};T_{\kappa} \right)=\mathbf{0} .
	\end{equation}
	The explicit formula of $G_{\tau}\left(\mathbf{C}^{(\kappa)};T_{\kappa}  \right)$ from (\ref{eq:stationary}) implies that $s^{(\kappa)}$ satisfies
	\begin{equation}\label{eq:s^k_solution}
		\mathcal{H}_{(\kappa)}s_{T_{\kappa}}^{(\kappa)}=\mathcal{G}_{(\kappa)} \mathbf{C}_{\overline{T}_{\kappa}}^{(\kappa)}-p^{(\kappa)}_{T_{\kappa}},
	\end{equation}
	where 
	\begin{equation} \nonumber
		\begin{aligned}
			\mathcal{H}_{(\kappa)} & = \bigtriangledown _{T_{\kappa}}^{2} \mathcal{D}_K\left( \mathbf{C}^{(\kappa)} \right) ,  \\
			\mathcal{G}_{(\kappa)}&=\bigtriangledown_{T_k,\overline{T}_{\kappa}}^{2}\mathcal{D}_K\left( \mathbf{C}^{(\kappa)} \right), \\
			p^{(\kappa)}&= \bigtriangledown_{T_k} \mathcal{D}_K\left( \mathbf{C}^{(\kappa)}\right).
		\end{aligned}
	\end{equation}
	And its solution satisfies
	\begin{equation}\label{eq:p^ks^k_solution}
		\left< p^{(\kappa)}_{T_{\kappa}},s_{T_{\kappa}}^{(\kappa)} \right> \le -\delta \lVert s^{(\kappa)}  \rVert ^2+ \lVert \mathbf{C}_{\overline{T}_{\kappa}}^{(\kappa)} \rVert ^2 / \left(4 \tau \right),
	\end{equation}
	where $\delta \in \left( 0,\min \left( 1, \ell\right) \right)$.
	$0<\tau  \le \frac{2\overline{\alpha }\delta \beta}{(\gamma+1)L^2}$, $\beta \in \left( 0,1 \right) $, 
	\begin{equation}
		\overline{\alpha }:=\min \left\{ \frac{1-2\sigma}{L/\delta-\sigma},\frac{2\left( 1-\sigma \right) \delta}{L},1 \right\},
	\end{equation}
	where $\sigma \in \left( 0,1/2 \right) $.
	
	If \eqref{eq:s^k_solution} is solvable and \eqref{eq:p^ks^k_solution} is satisfied, then update
	$s^{(\kappa)}$ using the Newton direction, as indicated by \eqref{eq:s^k_solution} and $s_{\overline{T}_{\kappa}}^{(\kappa)}=-\mathbf{C}_{\overline{T}_{\kappa}}^{(\kappa)}$.
	Otherwise, we use the gradient descent method to update the direction $s^{(\kappa)}$, that is, $s_{T_{\kappa}}^{(\kappa)} = -p^{(\kappa)}_{T_{\kappa}} ,
	s_{\overline{T}_{\kappa}}^{(\kappa)} = -\mathbf{C}_{\overline{T}_{\kappa}}^{(\kappa)}$. 

	
	Subsequently, we employ the line search methodology to ascertain the step size $\rho^{(\kappa)}$ for which the inexact Armijo-Goldstein criterion  \cite{Liu2020Optimization} can be adopted, $\mathbf{C}^{(\kappa+1)}=\mathbf{C}^{(\kappa)}+\rho^{(\kappa)} s^{(\kappa)}$, where
	\begin{align}\nonumber 
		\mathbf{C}^{(\kappa+1)} & := \left[ \begin{array}{c}
			\mathbf{C}_{T_{\kappa}}^{(\kappa)}+\rho^{(\kappa)} s_{T_{\kappa}}^{(\kappa)}\\
			\mathbf{C}_{\overline{T}_{\kappa}}^{(\kappa)}
			+s_{\overline{T}_{\kappa}}^{(\kappa)}\\
		\end{array} \right] \nonumber \\
		&=\left[ \begin{array}{c}
			\mathbf{C}_{T_{\kappa}}^{(\kappa)}+\rho^{(\kappa)} s_{T_{\kappa}}^{(\kappa)}\\
			0\\
		\end{array} \right] . \nonumber 
	\end{align}
	
	The parameter settings can adhere to the following conditions: 
	$0<\lambda \le \underline{\lambda}$, where
	\begin{equation}
		\underline{\lambda }\triangleq \underset{i}{\min}\left\{ \frac{\tau}{2}\left| \bigtriangledown _i \mathcal{D}_K\left( \mathbf{0} \right) \right|^2:\bigtriangledown _i \mathcal{D}_K\left( \mathbf{0}\right) \ne 0 \right\}.
	\end{equation}
	
	\textit{Halting conditions:}
	It is reasonable to terminate the algorithm at the $\kappa$-th step if $\kappa$ reaches the maximum number of iterations or $\mathbf{C}^{(\kappa)}$ satisfies $\text{supp}\left( \mathbf{C}^{(\kappa)} \right) \subseteq T_{\kappa}=T_{\kappa-1}$ and $\lVert G_{\tau_{\kappa}}\left( \mathbf{C}^{(\kappa)};T_{\kappa} \right) \rVert \le 10^{-6}$. $\text{supp}\left(\mathbf{C}^{(\kappa)} \right) $ be its support set consisting of the indices of the nonzero elements of $\mathbf{C}^{(\kappa)}$.
	
	\begin{remark}
		In comparing the ORLS solver limiting the polynomial order with the hybrid Newton algorithm for $\ell_0$-regularized polynomial optimization, it can be seen that the latter 
		does not rely on linear or linearized measurement as the former does but still needs $\mathcal{D}_K(\mathbf{C})$ strongly convex and needs to set four parameters $\sigma, \beta, \delta, \tau$ properly. For nonconvex data fitting models, measurement linearization/conversion as mentioned in Section \ref{sec:nonlinearMeasurement} that will result in strongly convex $\mathcal{D}_K(\mathbf{C})$ is practically considerable.   
	\end{remark}

	\begin{figure}[htbp]
		\centerline{\includegraphics[width=0.9\columnwidth]{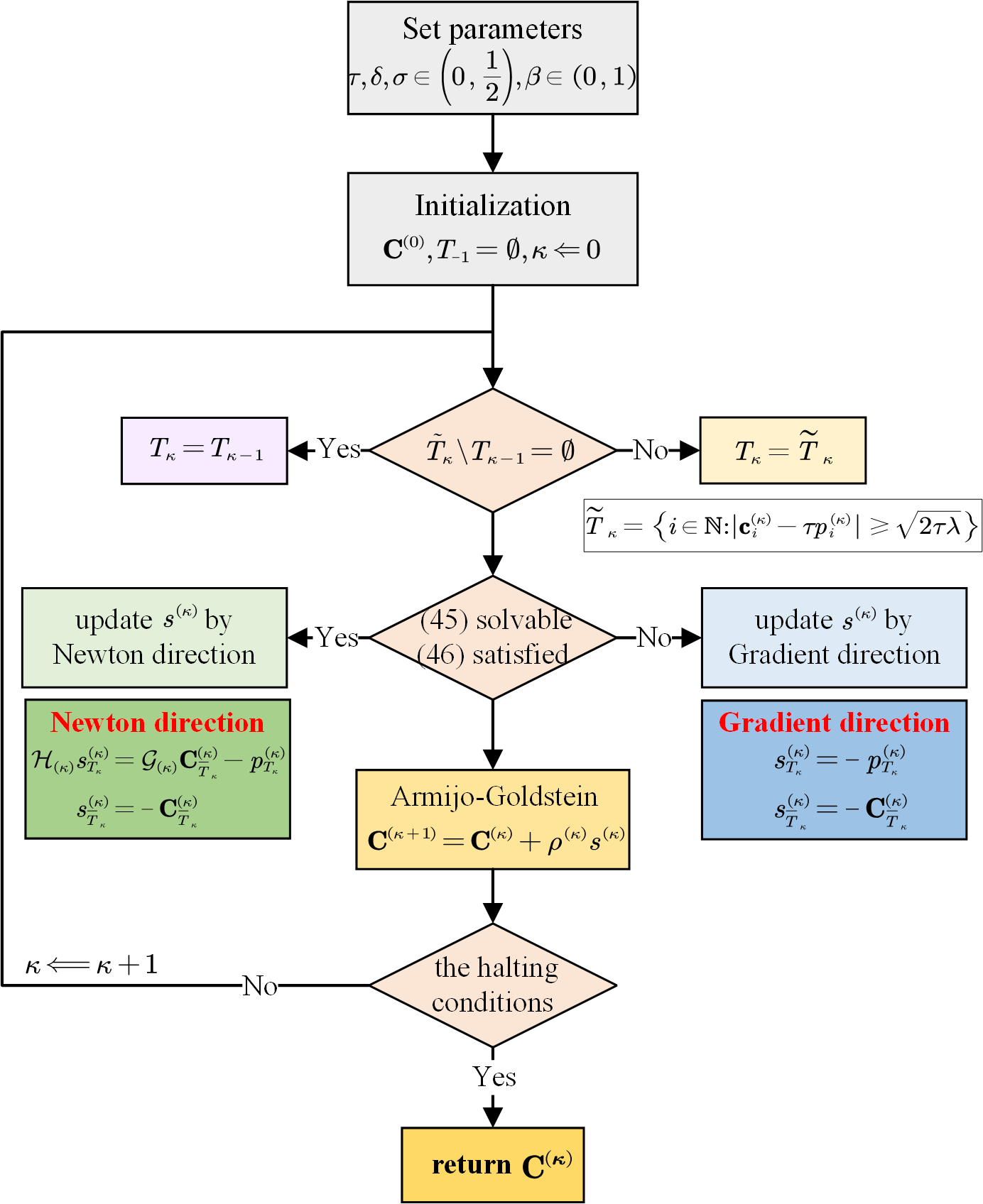}}
		\caption{Flowchart of the hybrid Newton algorithm for T-FoT $\ell_0$ optimization.}
	\end{figure}

		\section{Extension to Multiple Target Case} \label{sec:extension-MTT}
		Given that the measurement-to-track (M2T) association can be properly resolved, the multiple target version of Problems 2 and 3 can be the same addressed for each target in parallel. So, both proposed T-FoT optimization solvers can be directly extended to multiple T-FoT optimization in parallel. All the theoretical results obtained so far hold for each target.  
		
		Let us index measurement received at time $t$ by $j \in \mathcal{J}_t, t \in [k',k]$ which may be either the real measurement of the target or the clutter, and the target T-FoT set $\mathcal{I}$ specified by a group of parameters $\big\{ \mathbf{C}_{(i)}\big\}_{i\in \mathcal{I}}$, where $\mathbf{C}_{(i)}$ represents the parameters of the $i$-th estimated trajectory which can be of different orders with each other. 
		Here, we propose a joint M2T association and multi-trajectory fitting framework as follows.
		\begin{align}
			& \big\{ \mathbf{C}_{(i)} \big\} _{i\in \mathcal{I}} = \underset{\left\{ \mathbf{C}_{(i)} \right\} _{i\in \mathcal{I}}}{\text{argmin}} \sum_{i\in \mathcal{I}} \Big(   \lambda_i \lVert \mathbf{C}_{(i)}\rVert_0 +  \nonumber \\
			&~~~~~~ \sum_{t=k'}^k \sum_{j\in \mathcal{J}_t} \mathbf{1}_i \big( \mathbf{y}_{t}^{(j)}  \big)  \big\lVert \mathbf{y}_{t}^{(j)}-h_t\big( F( t;\mathbf{C}_{(i)} ) ,\bar{\mathbf{v}}_{t}^{(j)} \big) \big\rVert^2_{\text{var}(\mathbf{y}_t^{(j)})}  \Big)  ,\label{eq:MTT_C_k} \\
			& s.t. \ \ \mathbf{1}_i \big( \mathbf{y}_{t}^{(j)} \big) \in \left\{ 0,1 \right\}, \forall i\in \mathcal{I}, j\in \mathcal{J}_t, t\in [k',k] \nonumber \\
			& ~~~~~~ \sum_{i\in \mathcal{I}}{\mathbf{1}_i \big( \mathbf{y}_{t}^{(j)} \big)}\le 1, \forall j\in \mathcal{J}_t,  t\in [k',k] \label{eq:i_I-const} \\
			& ~~~~~~ \sum_{j\in \mathcal{J}_t} {\mathbf{1}_i \big( \mathbf{y}_{t}^{(j)} \big)}\le 1, \forall i\in \mathcal{I} , t\in [k',k]\label{eq:j_J-const} 
		\end{align}
		where $\lambda_i$ is the constraint coefficient of the $i$-th polynomial T-FoT, $\mathbf{1}_i \big( \mathbf{y}_{t}^{(j)} \big)$ represents the association of the $j$-th measurement received at time $t$ with the $i$-th trajectory, $\mathbf{1}_i \big( \mathbf{y}_{t}^{(j)} \big) =1$ for associated and $\mathbf{1}_i \big( \mathbf{y}_{t}^{(j)} \big) =0$ otherwise, \eqref{eq:i_I-const} is due to the constraint that each measurement can at most be associated to one trajectory while \eqref{eq:j_J-const} limits that each trajectory can at most be associated with one measurement at any time instant. 
		
		

		\begin{remark}
			The multi-target T-FoT model \eqref{eq:MTT_C_k} aims to solve the problem of data association and trajectory fitting jointly, which is nontrivial especially when the number of targets is unknown and false/missing data are presented. The optimization becomes more challenging when the order of each T-FoT is not given in advance but has to be taken into account in the cost function. 
			We leave these implementation problems aside but focus on the fundamental problem of comparing and evaluating a given set of T-FoT estimates based on the available ground truth. 
		\end{remark}

		\begin{remark}
			The above formulation for joint M2T association and multi-target trajectory fitting is intractable. Our approach solves it by separating the M2T association and fitting problems, and thus cannot guarantee optimality. 
			In fact, the M2T association itself, especially in the presence of clutter, is intractable, which has been a long-standing challenge in the tracking community \cite{Stefano2022DA}. The question of how to efficiently solve this joint optimization problem remains open.  
		\end{remark}

		\section{Simulation}\label{sec:simulation}
		We will consider both single and multiple target tracking scenarios in our simulation. 
		In the former, the real target trajectory is generated using the traditional SSM (and so the real trajectory is a time-series of discrete points, not a real continuous-time T-FoT) while in the latter, the real trajectories are just generated by polynomial T-FoTs. In both simulation scenarios, the proposed T-FoT approaches including the $\gamma$-limiting ORLS and $\ell_0$-hybrid Newton algorithms are carried out with a sliding time window of maximum $T_\text{w}=10$ sampling steps (corresponding to 1 second in total). The parameters needed in the hybrid Newton approach are set as given in Table \ref{tab:table1}. For comparison, the polynomial T-FoTs of fixed order $\gamma =1$ and $\gamma =2$ are also considered. 
		Furthermore, 
		comparison methods also include the $\ell_1$-ADMM algorithm which substitutes the $\ell_0$-norm by $\ell_1$-norm, namely $\varOmega _F\left( \mathbf{C}_{\gamma} \right) \triangleq  \lVert \mathbf{C}_{\gamma} \rVert _1$ in \eqref{eq:Const-T-FoT-Op}, and then adopts the ADMM solver \cite{lou2018fast} to estimate the T-FoT. In addition, the KF is also simulated in the single-target tracking case. 

		\subsection{Single Maneuvering Target Tracking}
		This simulation scenario is analogous to the one presented in Section 4.1.4 of \cite{hartikainen2008optimal}, where the motion of a maneuvering object switches between Wiener process velocity (WPV) with a low process noise of power spectral density $0.1$, and Wiener process acceleration (WPA) 
		with a high process noise of power spectral density $1$. The system is simulated with $100$ sampling steps, each of $1$s. 
		The real target motion model was manually set to WPV during steps $1$s-$30$s, $46$s-$70$s and $86$s-$100$s and to WPA during steps $31$s-$45$s, and $71$s-$85$s. This leads to four times of maneuvering. 
		
		We consider the linear measurement with additive, zero-mean noise, namely 
		\begin{equation}\label{eq:LinearMeansure-sim}
			\mathbf{y}_k= \mathbf{x}_k + \mathbf{v}_k,
		\end{equation}
		where the measurement noise is white Gaussian satisfying 
		\begin{equation} 
			\mathrm{E}\left[ {\mathbf{v}_k} \right] = 0, ~ \mathrm{E}\left[ {\mathbf{v}_k\mathbf{v}_j^\text{T}} \right] = \left[ {\begin{array}{*{20}{c}}
					{100}&0\\
					0&{100}
			\end{array}} \right]{\delta _{kj}},
		\end{equation}
		where \(\delta _{kj}\) 
		is equal to one if \(k=j\) and to zero otherwise.

		
		\begin{table}[!t]
			\caption{Setting of Parameters in hybrid Newton method\label{tab:table1}}
			\centering
			\begin{tabular}{ c | c}
				\hline\noalign{\smallskip}
				Parameters& Value Used \\ \noalign{\smallskip}\hline \noalign{\smallskip}
				$\sigma$  &  $5\times 10^{-5}$    \\ \noalign{\smallskip}\hline\noalign{\smallskip}
				$\beta$ &  0.5  \\ \noalign{\smallskip}\hline\noalign{\smallskip}
				$\delta$ &  $10^{-10}$     \\ \noalign{\smallskip}\hline\noalign{\smallskip}
				$\tau$&  1    \\ \noalign{\smallskip}\hline
			\end{tabular}
		\end{table}
		
		The simulation is carried out for $50$ Monte Carlo runs, where each run involves generating a trajectory randomly originating from the same initial point but different process noises and observation series, according to the above statistical models. The real trajectory and T-FoT estimates in one run are given in Fig. \ref{fig1}. 
		To gain further insights, the average root mean square error (RMSE) of the position estimation over time is given in Fig. \ref{fig2} and the time-averaged RMSE (over 100 sampling steps), as well as the average computing time for each step, are given in Table \ref{tab:table2}. The simulation is conducted on MATLAB R2018b. 
		
		\begin{itemize}
			\item  On the estimation accuracy, 
			the $\gamma$-limiting ORLS algorithm  performs the best, the proposed $\ell_0$-hybrid Newton approach the second, both outperforming the fitting using fixed-order (whether $\gamma=1$ or $\gamma=2$) and the $\ell_1$-ADMM approach. 
			The results demonstrated the importance of model adaption (the order of the T-FoT) in case of target maneuvering. It also confirms that $\ell_1$ regularization which may result in a high order yet minimized sum of the polynomial coefficients is undesired in our case. 
			
			\item  On the computing speed, all the three regularized T-FoT approaches (using adaptive orders) are unsurprisingly slower than those using fixed order. In particular, the hybrid Newton approach for $\ell_0$-regularized optimization suffers from the slowest calculation speed. It is valuable yet challenging to speed up this algorithm for practical use. The ORLS solver with limiting polynomial order $\gamma$ offers much superior computational speed compared to the direct hybrid Newton algorithm. While the $\ell_1$-ADMM approach runs fast too, close to the $\gamma$-limiting ORLS approach, its accuracy is further reduced. 
		\end{itemize}

		\begin{figure}[htbp]
			\vspace{-2mm}
			\centerline{\includegraphics[width=0.8\columnwidth]{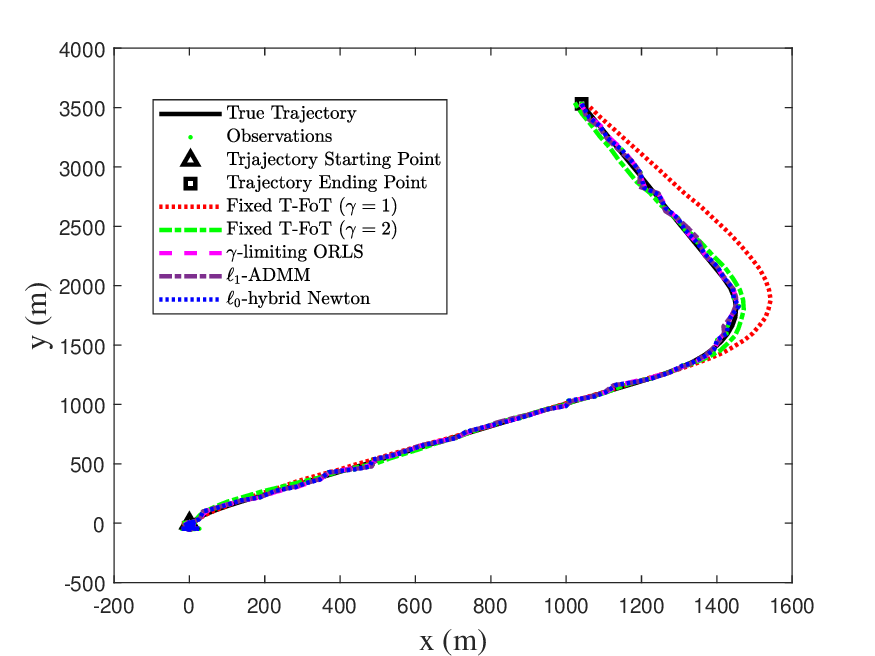}}
			\caption{Real trajectory and estimates given by different  estimators in one trial.}
			\label{fig1}
			\vspace{-2mm}
		\end{figure}

		\begin{figure}[htbp]
			\vspace{-2mm}
			\centerline{\includegraphics[width=0.8\columnwidth]{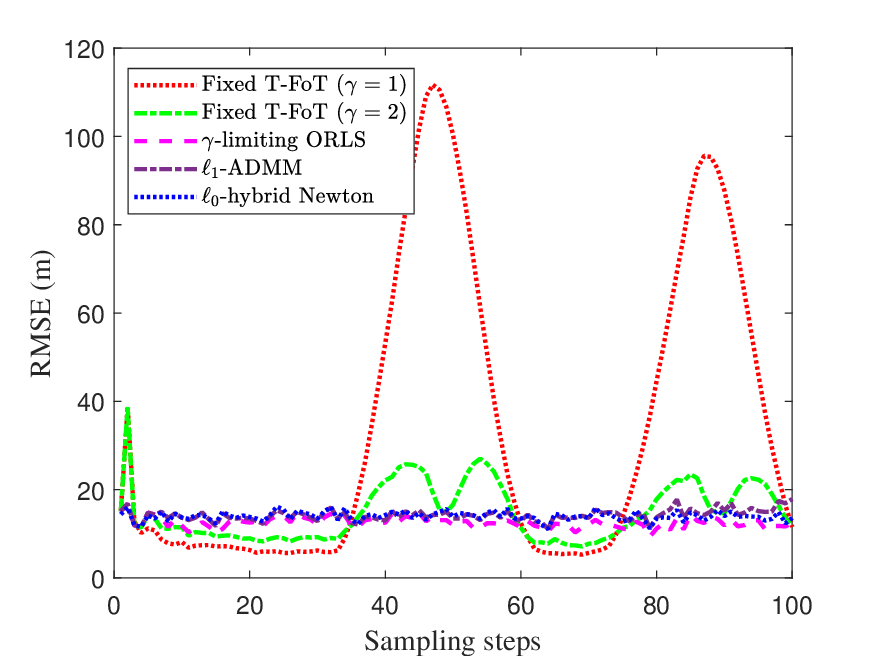}}
			\caption{RMSE of different  estimators over sampling steps.}
			\label{fig2}
			\vspace{-2mm}
		\end{figure}
		
		
		\begin{table}[!t]
			\caption{Average performance of different estimators over 100 sampling steps in the single target case \label{tab:table2}}
			\centering
			\begin{tabular}{ c | c c}
				\hline\noalign{\smallskip}
				Estimators    & Aver. RMSE (m) & Aver. Time (s) \\ \noalign{\smallskip}\hline \noalign{\smallskip}
				Fixed order $\gamma =1$   &  34.1762  & 0.0409     \\ \noalign{\smallskip}\hline\noalign{\smallskip}
				Fixed order $\gamma =2$  & 14.7818   & 0.0389 \\ \noalign{\smallskip}\hline\noalign{\smallskip}
				$\ell_1$-ADMM   &  14.2232 & 0.0644    \\ \noalign{\smallskip}\hline\noalign{\smallskip}
				$\gamma$-limiting ORLS  & \textbf{12.5376}   & 0.0821 \\ \noalign{\smallskip}\hline\noalign{\smallskip}
				$\ell_0$-hybrid Newton  & 13.9804 & 22.5460     \\ \noalign{\smallskip}\hline 
			\end{tabular}
		\end{table}

		\subsection{Multiple Target Tracking}
		We now consider a more challenging tracking scenario with two targets, each with a lifespan from time $1$s to time $100$s. One target exhibits maneuvers at times $20$s and $70$s, whereas the other target maneuvers around time $40$s. The trajectories that are generated using polynomial curves (for which a maneuver occurs when the polynomial order changes) in one run are given in Fig. \ref{fig:multi_scen}. Similar with \eqref{eq:LinearMeansure-sim}, measurements are made on the 2-dimensional position with zero-mean Gaussian noise of standard deviation $1$m. The two targets are so well separated that the position measurements and the tracks can be correctly associated by using the standard global nearest-neighbor approach. 
		The target detection probability $P_D=1$. Clutter follows a Poisson model with an average $15$ clutter-points per measuring time  
		in the region $\left[-170, 150\right]$m $\times \left[-150, 300\right]$m. 
		Thanks to the unit detection probability and the perfect M2T association, the T-FoT estimators are free of the MD problem. 
		Due to the absence of the target dynamics and the background profile, no Bayesian filters can be properly set up, and so our comparison includes only T-FoT solvers using different polynomials. 
		
		For the performance evaluation, both the OSPA (optimal sub-pattern assignment) \cite{schuhmacher2008consistent} and the Star-ID (spatio-temporal-aligned trajectory integral distance) \cite{Li25TFoT-part1} metrics are used. The former compares the real states with the point state estimates extracted from the estimated T-FoT at each measuring time while the latter calculates the integral distance between the estimated and real T-FoTs in the fitting time-window part which is basically a cumulative distance of trajectories over the concerning time-window. We set the cutoff parameter $c= 20$m for the OSPA and relatively segment/trajectory cutoff parameter $c_\text{S}=c_\text{T}= 20$m for the Star-ID, and the same metric order $p=2$ for both of them.  
		The average OSPA error and Star-ID of different T-FoT estimators over time are shown in Fig. \ref{fig:ospa} and Fig. \ref{fig:Star-ID}, respectively. 
		Furthermore, the time-averaged Star-ID (TA-Star-ID \cite{Li25TFoT-part1}) of different T-FoT estimators which divides the Star-ID by the length of the concerning time-window is shown in Fig. \ref{fig:TA-Star-ID}. 
		The averages of the TA-Star-ID, of the OSPA distances and of the computing time over 100 sampling steps are given in Table \ref{tab:table3}.

		\begin{figure}[htbp]
			\centerline{\includegraphics[width=0.85\columnwidth]{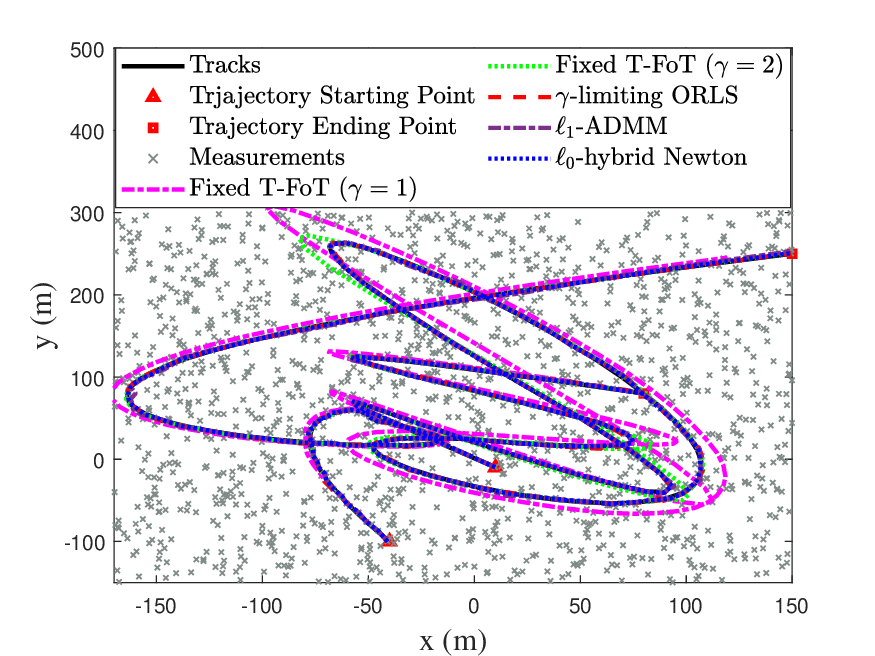}}
			\caption{Real and estimated T-FoTs of two targets in the cluttered environment.}
			\label{fig:multi_scen}
		\end{figure}
		
		\begin{figure}[htbp]
			\centerline{\includegraphics[width=0.8\columnwidth]{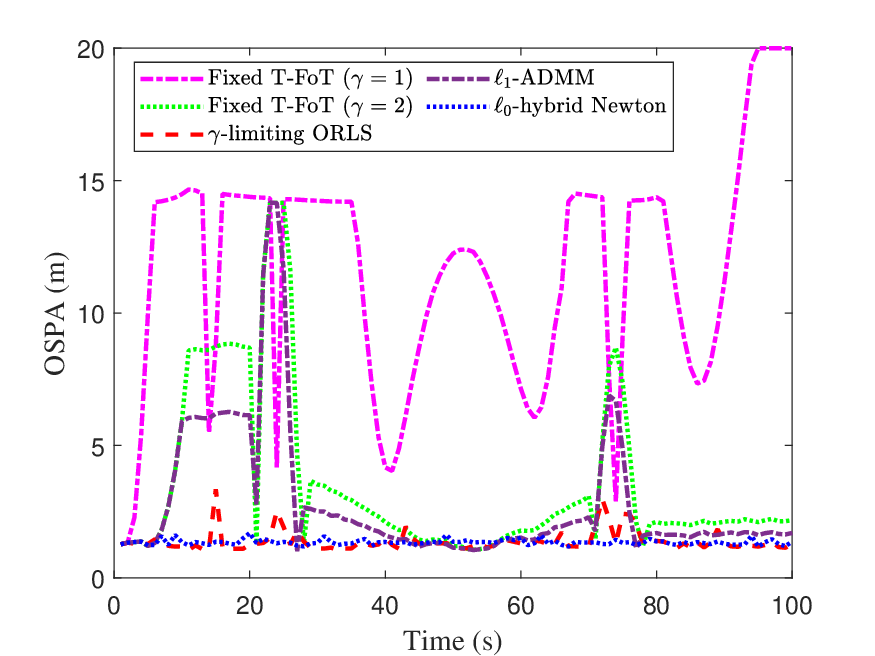}}
			\caption{Average OSPA error of different T-FoT estimators over time.}
			\label{fig:ospa}
		\end{figure}
		
		\begin{figure}[htbp]
			\centerline{\includegraphics[width=0.8\columnwidth]{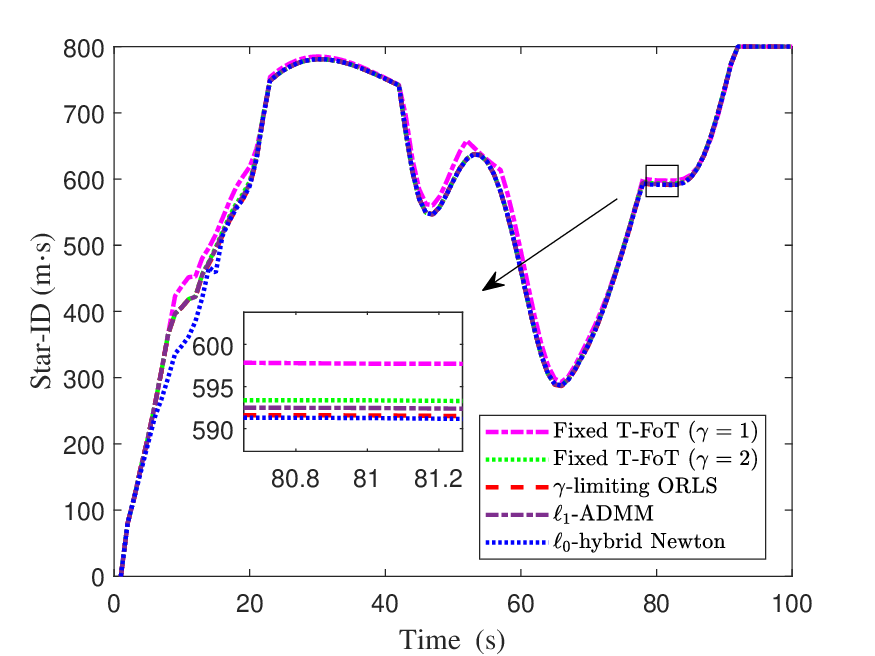}}
			\caption{Average Star-ID error of different T-FoT estimators over time.}
			\label{fig:Star-ID}
		\end{figure}
		
		\begin{figure}[htbp]
			\centerline{\includegraphics[width=0.8\columnwidth]{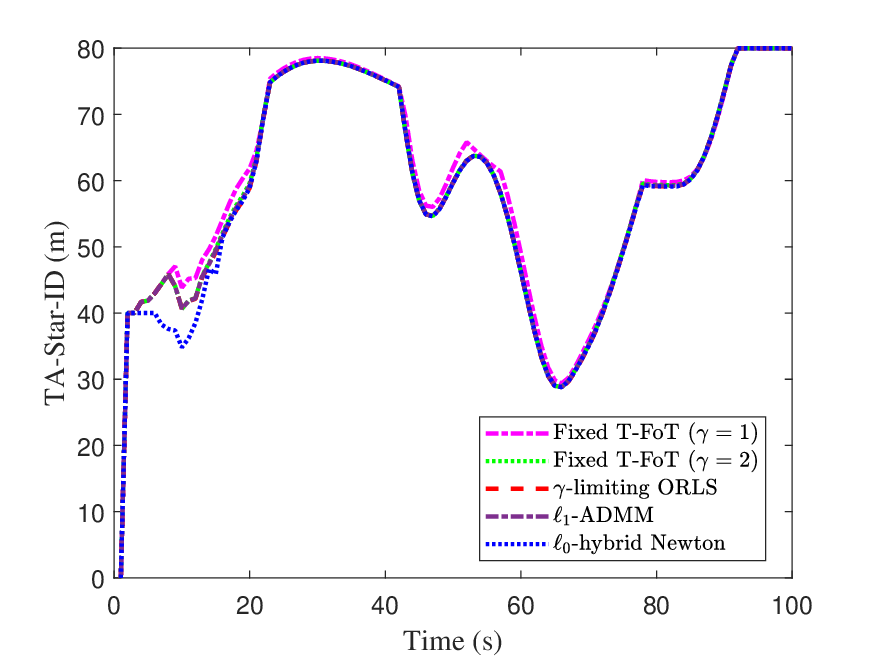}}
			\caption{Average TA-Star-ID error of different T-FoT estimators over time.}
			\label{fig:TA-Star-ID}
		\end{figure}
		
		
		
		\begin{table*}[!t]
			\caption{Average performance of different estimators over 100 simulation steps in the multi-target case \label{tab:table3}}
			\centering
			\begin{tabular}{ c | cccc}
				\hline\noalign{\smallskip}
				T-FoT solver & OSPA (m) & TA-Star-ID (m) & Star-ID (m$\cdot$s) & Time (s) \\ \noalign{\smallskip}\hline \noalign{\smallskip}
				Fixed order $\gamma =1$ &11.6494 & 59.5925 & 580.8893 &0.0817     \\ \noalign{\smallskip}\hline\noalign{\smallskip}
				Fixed order $\gamma =2$ &3.5599 & 58.7122 & 572.1178 &\textbf{ 0.0812} \\ \noalign{\smallskip}\hline\noalign{\smallskip}
				$\ell_1$-ADMM & 2.8418 & 58.6648 &  571.6438 & 0.1254    \\ \noalign{\smallskip}\hline\noalign{\smallskip}
				$\gamma$-limiting ORLS & 1.3761 & 58.6200 & 571.1960 &  0.1499 \\ \noalign{\smallskip}\hline\noalign{\smallskip}
				$\ell_0$-hybrid Newton & \textbf{1.3447} & \textbf{58.1211} & \textbf{566.9380} & 27.6312     \\ \noalign{\smallskip}\hline
			\end{tabular}
		\end{table*}

		Regarding the estimation accuracy, it is evident that in terms of both OSPA distance and Star-ID, the polynomial fitting using the $\ell_0$-hybrid Newton approach performs the best again in all, the $\gamma$-limiting ORLS the second and the $\ell_1$-ADMM the third, all outperforming significantly the fitting with a fixed order. 
		However, the hybrid Newton method exhibits a considerably much higher computing cost compared to the others, making it impractical for real-time scenarios. In contrast, the slight increase of the computational cost in the proposed $\gamma$-limiting ORLS approach together with its outstanding accuracy make it well-suited for real-time applications. 
		
		In comparing Fig. \ref{fig:TA-Star-ID} with Fig. \ref{fig:ospa}, the simulation also confirms that the TA-Star-ID performs consistently with the OSPA metric while the former, which accounts for the whole trajectory-segment in the time-window, varies much smoother than the latter in the time series. This precisely aligns with the remarkable distinction between the trajectory estimator and the point state estimator.

		\section{Conclusion}\label{sec:conclusion}
		Our series of work aims to establish a data-driven SP-based learning-for-tracking framework which yields the continuous-time trajectory rather than discrete-time point estimates. This paper, as part II of the series, addresses the learning of the trajectory SP trend by polynomial, which plays a key role in the whole SP learning \cite{Li25TFoT-part3}. In particular, we address the polynomial T-FoT optimization with two distinctive strategies of regularization, balancing the fitting accuracy and the polynomial smoothness/simplicity. One greedily searches the optimal order $\gamma$ of the polynomial in a narrow, bounded range. In particular, the $\gamma$-limiting ORLS solver that fits the linear measurement model is detailed. The other employs a hybrid Newton approach to address the NP-hard $\ell_0$ regularization with convex data fitting model. 
		The simulation results demonstrate that our proposed T-FoT optimization method significantly outperforms fixed-order polynomial fitting approaches, as well as the $\ell_1$ regularization solved by the ADMM method, in terms of tracking accuracy in the case of both single and multiple maneuvering target tracking. While the computing speed of $\gamma$-limiting ORLS solver is acceptable, the hybrid Newton approach for $\ell_0$ regularization is computationally intensive and not suitable for online tracking. More effort is needed to solve the regularized optimization in real time and integrate its learning with that of the RSP. 


\end{document}